\documentclass[acmlarge,nonacm]{aamas} 

\usepackage{balance} 
\usepackage{xspace}
\usepackage{graphicx}
\usepackage{calc}
\usepackage{amsmath}
\usepackage{amsthm}
\usepackage{color}
\usepackage[capitalise,noabbrev]{cleveref}
\usepackage[disable]{todonotes}
\usepackage{enumitem}
\usepackage{subcaption}
\usepackage{soul}
\usepackage{textcomp}
\usepackage{tikz}
\usepackage{stackengine}
\usepackage{marvosym}
\usepackage{tcolorbox}
\usetikzlibrary{automata, arrows.meta, 
positioning, decorations.pathreplacing, calligraphy, shapes.geometric}
\usepackage{varwidth}

\setcopyright{ifaamas}
\acmConference[AAMAS '25]{Proc.\@ of the 24th International Conference
on Autonomous Agents and Multiagent Systems (AAMAS 2025)}{May 19 -- 23, 2025}
{Detroit, Michigan, USA}{A.~El~Fallah~Seghrouchni, Y.~Vorobeychik, S.~Das, A.~Nowe (eds.)}
\copyrightyear{2025}
\acmYear{2025}
\acmDOI{}
\acmPrice{}
\acmISBN{}

\acmSubmissionID{}

\title[AAMAS-2025 Formatting Instructions]{$\WR$obin Hood Reachability Bidding Games\\ Full Version}

\author{Shaull Almagor}
\affiliation{
  \institution{Technion}
  \city{Haifa}
  \country{Israel}}
\email{shaull@technion.ac.il}

\author{Guy Avni}
\affiliation{
  \institution{University of Haifa}
  \city{Haifa}
  \country{Israel}}
\email{gavni@cs.haifa.ac.il}

\author{Neta Dafni}
\affiliation{
  \institution{Technion}
  \city{Haifa}
  \country{Israel}}
\email{netad@campus.technion.ac.il}

\begin{abstract}
Two-player graph games are a fundamental model for reasoning about the interaction of agents. These games are played between two players who move a token along a graph. In bidding games, the players have some monetary budget, and at each step they bid for the privilege of moving the token. 
Typically, the winner of the bid either pays the loser or the bank, or a combination thereof. 
We introduce Robin Hood bidding games, where at the beginning of every step the richer player pays the poorer a fixed fraction of the difference of their wealth. After the bid, the winner pays the loser. Intuitively, this captures the setting where a regulating entity prevents the accumulation of wealth to some degree. 

We show that the central property of bidding games, namely the existence of a threshold function, is retained in Robin Hood bidding games. We show that finding the threshold can be formulated as a Mixed-Integer Linear Program.
Surprisingly, we show that the games are not always determined exactly at the threshold, unlike their standard counterpart. 
\end{abstract}

\keywords{Bidding Games, Discounting, Reachability Games, Wealth Regulation}

\newcommand{\BibTeX}{\rm B\kern-.05em{\sc i\kern-.025em b}\kern-.08em\TeX}

\DeclareMathOperator*{\argmax}{\arg\!\max}
\DeclareMathOperator*{\argmin}{\arg\!\min}

\newcommand{\icol}[1]
{\left(\begin{smallmatrix}#1\end{smallmatrix}\right)}
\newcommand{\cG}{\mathcal{G}}

\newcommand{\cI}{\mathcal{I}}
\newcommand{\Th}{\mathsf{Th}}
\newcommand{\WR}{
\stackMath %
\stackon[-0.5ex]{\mathfrak{R}}{\hspace{1pt}\includegraphics[width=2ex]{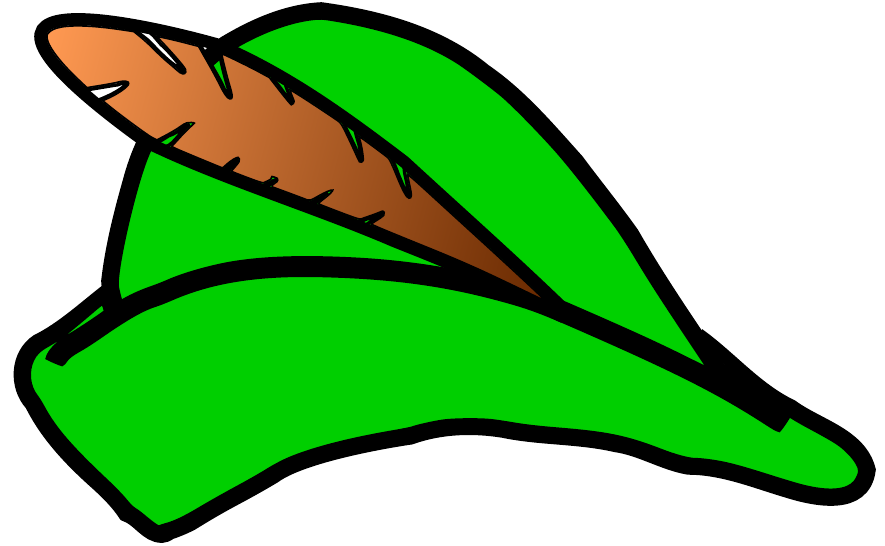}}
}
\newcommand{\WRinv}{
\stackMath %
\stackon[-0.8ex]{\mathfrak{R}^{-1}}{\hspace{-8.5pt}\includegraphics[width=2ex]{robinhood1.pdf}}
}

\renewcommand{\vec}[1]{\boldsymbol{#1}}
\newcommand{\bbR}{\mathbb{R}}
\newcommand{\bbN}{\mathbb{N}}

\newcommand{\tup}[1]{\langle #1 \rangle}

\newcommand{\play}{\texttt{play}}
\newcommand{\neig}{\Gamma}

\newtheorem{remark}{Remark}[section]

\newcommand{\eVecOne}{\vec{E^\mathtt{vec}_1}}
\newcommand{\eVecTwo}{\vec{E^\mathtt{vec}_2}}
\newcommand{\eValOne}{E^\mathtt{val}_1}
\newcommand{\eValTwo}{E^\mathtt{val}_2}

\newcommand{\xinit}{x_\mathrm{init}}

\newcommand{\pre}{\mathrm{pre}}

\newcommand{\unravel}{\mathcal{D}}

\newcommand{\PO}{\Ladiesroom\MVOne\xspace}
\newcommand{\PT}{\Gentsroom\MVTwo\xspace}

\begin{document}


\pagestyle{fancy}
\fancyhead{}


\maketitle 

\section{Introduction}
\label{sec:intro}
A \emph{reachability game} is a 2-player game played on a graph, by placing a token on one of the vertices and moving it along the edges according to some predefined rules, where the goal of Player 1 (denoted \PO) is to reach a set of target vertices, and the goal of Player 2 (denoted \PT) is to prevent that. Reachability games are fundamental in automated synthesis of systems~\cite{pnueli1989synthesis}, where a system plays against an environment (e.g., controller synthesis~\cite{girard2012controller}, robotic planning~\cite{fainekos2006translating}, network routing~\cite{brihaye2019dynamics}, etc.).

In \emph{bidding games}~\cite{lazarus1996richman,lazarus1999bidding, avni2019infiniteDuration,avni2020CONCUR}, each player has a \emph{budget} (a real value in $[0,1]$, where the sum of the budgets is assumed to be normalized to $1$) at any given moment, and the movement of the token in each step is determined by an \emph{auction}, resulting in the higher bidder moving the token. 
We focus on \emph{Richman} games~\cite{lazarus1996richman}, where the winner of the auction pays their bid to the loser (see e.g.,~\cite{avni2021bidding} for other bidding mechanisms).

Bidding games are useful for modelling settings where agents compete for some resource (e.g., money, computational resources, etc.) and use these resources to direct the interaction. 


A common phenomenon in bidding games is that if one of the players accumulates a high-enough portion of the budget, that player can 
force the game to reach any desired location. In loose terms, ``the rich can do whatever they want''.
In some settings this phenomenon is desirable, e.g., when modelling the interaction between an attacker and defender of a security system, and the budget is computational resources -- nothing prevents either party from hogging resources in order to win.
In many other settings, however, the players operate under some regulating entity (e.g., a scheduler in an operating system, or monetary regulation), which prevents the accumulation of excessive wealth in order to achieve some fairness, or to inspire active participation in the game.

A standard means to regulate wealth is to redistribute some of the wealth of the rich to the poor, \`a la Robin Hood's \emph{steal from the rich and give to the poor}~\cite{poddar2012exploring,franko2013inequality}. Note that this is not the same as taxation in that in standard taxation it is not at all clear that taxes go to the poor, nor is it the case that the poor are not taxed.

In this work, we introduce a variant of bidding games called \emph{Robin Hood bidding games} which incorporates wealth regulations. In a Robin Hood bidding game, each auction is preceded by a \emph{wealth redistribution} phase: the richer player pays the poorer player a constant fraction (denoted $\lambda$) of the difference between their budgets. The classical model of bidding games then corresponds to $\lambda=0$. We only consider $0\leq\lambda< \frac{1}{2}$, as $\lambda\ge \frac{1}{2}$ would mean that the richer player becomes the poorer (or equal, for $\lambda=\frac12$), which is of little motivation.
\begin{example}
    Consider the game depicted in \cref{fig: first example}, starting in $v_\mathrm{left}$, where the target for \PO is $v_1$. The wealth redistribution factor is $\lambda=\frac{1}{8}$. 

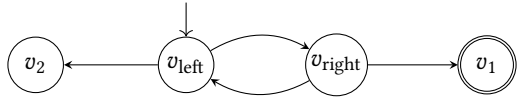
\begin{figure}[ht]
\centering  
\begin{tikzpicture}[auto,node distance=2.8cm,scale=1]
    \tikzset{every state/.style={minimum size=21pt, inner sep=1}};
    \node (vleft) [initial above, state, initial text = {}] at (0,0) {$v_\mathrm{left}$};
    \node (v2) [state] at (-2,0) {$v_2$};
    \node (vright) [state] at (2,0) {$v_\mathrm{right}$};
    \node (v1) [state, accepting] at (4,0) {$v_1$};
    \path [-stealth]
    (vleft) edge (v2)
    (vleft) edge [bend left] (vright)
    (vright) edge [bend left] (vleft)
    (vright) edge (v1);
\end{tikzpicture}
\caption{A Robin Hood game. The target for \PO is $v_1$.}
\label{fig: first example}
\end{figure}
\begin{figure}[ht]
\centering  
\includegraphics[width=0.75\columnwidth]{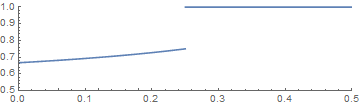}
\caption{Threshold of $v_{\textrm{left}}$ in \cref{fig: first example} as a function of $\lambda$.}
\label{fig: discontinuous threshold}
\end{figure}
    
    Observe that \PO wins if the play reaches $v_1$, and loses if it reaches $v_2$ or oscillates indefinitely between $v_\mathrm{left}$ and $v_\mathrm{right}$. 
    Recall that we assume the sum of budgets of the players is $1$. As we show in \cref{sec: example}, \PO needs a starting budget of at least $0.7$ in order to win. We demonstrate why \PO loses when starting with $0.6$. At a glance, the optimal strategies for the players induce the play in \cref{fig:optimal play intro example}. 
    
    Starting with budgets $\icol{0.6\\0.4}$, the first step is to apply wealth redistribution (WR, for short). Since $\lambda=\frac18$ and $0.6-0.4=0.2$, then \PO pays $0.025$, so the new budgets are $\icol{0.575\\0.425}$. Note that if \PO loses the bidding at $v_\mathrm{left}$, she loses the game. Therefore, she must bid at least $0.425$ (we use the common assumption that ties are broken in favor of \PO). Fortunately, she has sufficient budget for this bid. She moves the token to $v_\mathrm{right}$ and the new budgets are $\icol{0.15\\0.85}$ (see (1) in \cref{fig:optimal play intro example}).
    Next, WR is applied (with \PT paying). In $v_\mathrm{right}$ \PT must win the bidding, or he loses the game. To do so, he must bid strictly more than \PO. He bids $0.2375+0.001=0.2385$, wins the bid and moves the token to $v_\mathrm{left}$ (see (2)). Then, WR leaves \PT as the richer, and he wins the game (3) by out-bidding \PO and moving to $v_2$.
\end{example}

\begin{figure}[ht]
        \centering
        \small
        \[
         \underbrace{\icol{0.6\\0.4}}_{v_\mathrm{left}}
        \overset{\text{WR}}{\longrightarrow}
        \underbrace{\icol{0.575\\0.425}}_{v_\mathrm{left}}
        \overset{(1)}{\longrightarrow}
        \underbrace{\icol{0.15\\0.85}}_{v_\mathrm{right}}
        \overset{\text{WR}}{\longrightarrow}
        \underbrace{\icol{0.2375\\0.7625}}_{v_\mathrm{right}}
        \overset{(2)}{\longrightarrow} 
         \underbrace{\icol{0.476\\0.524}}_{v_\mathrm{left}}
        \overset{\text{WR}}{\longrightarrow}
        \underbrace{\icol{0.4805\\0.5195}}_{v_\mathrm{left}}
        \overset{(3)}{\longrightarrow}
        v_2
        \]
        \caption{A losing play for \PO.}
        \label{fig:optimal play intro example}
    \end{figure}

The central question in the study of bidding games is the existence of a \emph{threshold} function for the game: a function that assigns for each vertex $v$ a value $\Th(v)$, such that if \PO starts with budget more than $\Th(v)$ then she wins, and if she starts with less than $\Th(v)$, she loses. 
It is shown in~\cite{avni2021bidding} that every reachability bidding game has a threshold, and finding its value is in $\mathrm{NP}\cap\mathrm{coNP}$.


We extend the known results regarding the existence of thresholds for reachability bidding games, and show that every Robin Hood reachability bidding games has a threshold (\cref{sec: thresholds main}). We further show that computing this threshold can be done using Mixed Integer Linear Programming (MILP).
Additionally, unlike previous works, we discuss what happens when the initial budget equals exactly the threshold (\cref{sec: what happens at threshold}). We find, surprisingly, that the game might not be determined at the threshold (i.e., neither player has a winning strategy), a behavior that does not occur in standard bidding games where \PO wins ties
(nor in general turn based games without concurrent biddings~\cite{martin1975borel}). 
Apart from the result itself, we believe it is important to draw attention to discussions about the behavior at the threshold. Indeed -- it is often the case that optimal strategies work by reaching certain vertices exactly at their threshold. Still in \cref{sec: what happens at threshold}, we show that given the threshold function and a vertex $v$, we can decide in polynomial time if the game is undetermined at $v$ with budget $\Th(v)$, and if not -- who the winner is.

In addition to these contributions, we observe curious behavior of the threshold function when $\lambda$ is treated as a parameter. Specifically, in \cref{sec: example} we conduct an elaborate analysis of the example in \cref{fig: first example} and show that this function might be discontinuous (see~\cref{fig: discontinuous threshold}). We also demonstrate a toolbox for analyzing specific games when $\lambda$ is a parameter.



\section{Preliminaries}
\label{sec: definitions}
A \emph{graph} is $G=(V,E)$ where $V$ is a set of vertices and $E\subseteq V\times V$ is a set of edges. For $v\in V$ we denote by $\neig(v)=\{u\mid (v,u)\in E\}$ the set of \emph{neighbors} of $v$. If $\neig(v)=\emptyset$ then $v$ is a \emph{sink}.

A \emph{Robin Hood reachability bidding game} is $\cG=\tup{G,v_0,\xinit,\lambda,T}$, where $G=(V,E)$ is a finite graph, $v_0\in V$ is an initial vertex, $\xinit$ is \PO's initial budget, $\lambda\in[0,\frac{1}{2})$ is the \emph{wealth redistribution} factor, and $T\subseteq V$ is a set of \emph{target} vertices for \PO. 
We assume for convenience that the vertices in $T$ are sinks.
We sometimes omit $v_0$ and $\xinit$, when the discussion is not limited to specific initial vertex and budget.

The game is played between 2 players as follows. At each step, a token is placed on a vertex $v\in V$ (initially $v_0$), and each of the players has a \emph{budget}, the budgets being described by a vector $\vec{w}\in [0,1]^2$ (initially $\icol{\xinit\\1-\xinit}$). 
For clarity, we denote vectors in bold (e.g., $\vec{v}$).

The game proceeds in steps, each consisting of the following phases:
\begin{enumerate}
    \item Wealth Redistribution (abbreviated WR and denoted $\WR$): Each player's budget is updated using the operator
    \[
        \WR(x)=(1-2\lambda)x+\lambda
    \]
    \item Bidding: Each player (concurrently) makes a bid within their budget. The player with the higher bid $b$ wins the bidding (a tie is broken in favor of \PO) and pays $b$ to the other player. The budgets are updated accordingly.
    \item Moving: The player who wins the bidding moves the token to a neighbor of $v$ of their choice.
\end{enumerate}
We remark that $\WR$ can be viewed as the linear operator on vectors given by the matrix
\[ \WR=
    \begin{pmatrix}
        1-\lambda & \lambda \\
        \lambda & 1-\lambda
    \end{pmatrix}
\]
Indeed, in this view we have $\WR \icol{x \\ 1-x}=\icol{(1-2\lambda)x+\lambda \\ (1-2\lambda)(1-x)+\lambda}$. We abuse notation and use either view as convenient. 

A \emph{configuration} is a pair $(v,x)\in\left(V\times [0,1]\right)$ 
where $v$ is the current vertex and $x$ is \PO's budget (so the budget for \PT is $1-x$). 
A \emph{strategy}\footnote{Note that our definition is restricted to \emph{memoryless} strategies. Since we consider Reachability objectives, this is sufficient.} for \PO is a function $\sigma_1\colon V\times [0,1]\to[0,1]\times V$ describing for each configuration $(v,x)$ a bid $b\in [0,1]$ and a neighbor $u$ of $v$. That is, if $\sigma_1(v,x)=(b,u)$ then we require $b\le \WR(x)$ (as \PO's budget during the bidding is $\WR(x)$) and $(v,u)\in E$. 
A strategy $\sigma_2$ for \PT is defined similarly, changing the budget requirement to $b\le \WR(1-x)$, as $\WR(1-x)$ is \PT's budget. 
Given an initial configuration $(v_0,\xinit)$ and strategies $\sigma_1,\sigma_2$ for \PO and \PT respectively, their {induced play}, denoted $\play(\sigma_1,\sigma_2,v_0,\xinit)=(v_0,x_0),(v_1,x_1),\ldots$, is a (finite or infinite) sequence of configurations defined as per the steps above. Specifically, 
$x_0=x_{\mathrm{init}}$, and for every $n\ge 0$, if $v_n$ has no outgoing edges, the sequence terminates. Otherwise, consider $(b_i,u_i)=\sigma_i(v_0,\WR(x_n))$ for $i\in \{1,2\}$. Then, if $b_1\ge b_2$ we have $(v_{n+1},x_{n+1})=(u_1,\WR(x_n)-b_1)$ and if $b_1<b_2$ then $(v_{n+1},x_{n+1})=(u_2,\WR(x_n)+b_2)$.
If the play reaches $T$ then \PO \emph{wins} the play, and otherwise \PT wins. For $i\in \{1,2\}$, a strategy $\sigma_i$ is \emph{winning for Player $i$ from $(v_n,\xinit)$} if for every strategy $\sigma_{1-i}$ for Player $1-i$, the induced play is winning for Player $i$.

Observe that the initial budgets of the player sum to $1$. Moreover, their sum is maintained by the WR operation and after each bidding. Thus, their sum remains $1$ throughout the play. 

Given $\tup{G,\lambda,T}$,
 a \emph{threshold function} is a function $\Th\colon V\to[0,1]$ such that for every $(v,\xinit)$ the following holds for the game $\tup{G,v,\xinit,\lambda,\alpha}$:
\begin{itemize}[leftmargin=*]
    \item If $\xinit>\Th(v)$, \PO has a winning strategy from $(v,\xinit)$.
    \item If $\xinit<\Th(v)$, \PT has a winning strategy from $(v,\xinit)$.
\end{itemize}
We call $\Th(v)$ \emph{1-strong} if an initial budget of exactly $\Th(v)$ wins for \PO, and \emph{2-strong} if it wins for \PT. If neither is true (the game is undetermined at the threshold), we call the threshold \emph{weak}.

\section{An Enlightening Example}
\label{sec: example}

In this section we expand the discussion regarding the game in \cref{fig: first example}. This serves to gain familiarity with the model, but in addition -- enables us to prove certain interesting properties of Robin-Hood games, namely \cref{cor:threshold not continuous}. Moreover, the tools we present may be of use in analyzing other games (c.f., \cref{rmk:toolbox}).

We are interested in finding a threshold for $v_\mathrm{left}$ specifically, as a function of $\lambda$. We denote this threshold by $\tau(\lambda)$.

For $\lambda=0$, it is shown in~\cite{avni2021bidding} that $\frac{2}{3}$ is a 1-strong threshold. Intuitively, this budget allows \PO to ensure the play oscillates between $v_\mathrm{left}$ and $v_\mathrm{right}$, with her budget increasing with each move to $v_\mathrm{left}$, until it is high enough to allow her to win two consecutive biddings and reach $v_1$.
In the following, we show the existence of a threshold $\tau(\lambda)$ for every $\lambda$.

A play can end by either reaching $v_1$ or $v_2$ (and \PO wins or loses respectively), or oscillate infinitely between $v_\mathrm{left},v_\mathrm{right}$, in which case \PO loses. Regardless, a play can be described by a finite or infinite sequence of iterations, each comprising four phases: (1) WR in $v_\mathrm{left}$; (2) bidding in $v_\mathrm{left}$; (3) WR in $v_\mathrm{right}$; and (4) bidding in $v_\mathrm{right}$. This sequence can be finite and be followed by a move into $v_1$ or $v_2$ which ends the play, or be infinite. 
For the $n$'th iteration, we denote the budget vectors as follows. In $v_\mathrm{left}$: $\vec{w}^{(n)}_{\mathrm{left}}$ and $\vec{w}^{(n)}_{\mathrm{wr},\mathrm{left}}$, before and after WR, respectively; and similarly $\vec{w}^{(n)}_{\mathrm{right}}$ and $\vec{w}^{(n)}_{\mathrm{wr},\mathrm{right}}$ for $v_\mathrm{right}$.
We also denote the first entries of these vectors, that is, \PO's budgets, by $x^{(n)}_{\mathrm{left}},x^{(n)}_{\mathrm{wr},\mathrm{left}},x^{(n)}_{\mathrm{right}},x^{(n)}_{\mathrm{wr},\mathrm{right}}$.

\subsection{Alternative Tie Breaking}
\label{sec:example alternative tie breaking}
It is useful to first consider the case that ties are not always broken in favor of \PO, but instead in favor of \PO when in $v_\mathrm{left}$ and in favor of \PT when in $v_\mathrm{right}$ (resulting in moving from $v_\mathrm{left}$ to $v_\mathrm{right}$ and vice-versa).

We describe optimal strategies $\sigma_1,\sigma_2$ for the players and their resulting play.
At any iteration $n$, when in $v_\mathrm{left}$ and about to bid, if $x^{(n)}_{\mathrm{wr},\mathrm{left}}<\frac{1}{2}$ then \PT can bid $\frac12+\epsilon$ and move to $v_2$, so \PO instantly loses. For this reason we only consider initial budgets of at least $\frac{1}{2}$ for this example. Note that having a budget of at least $\frac{1}{2}$ before or after WR is equivalent, and the same for a strict inequality. That is, $x>\frac12$ if and only if $\WR(x)>\frac12$ and $\WR(\frac12)=\frac12$.
If \PO has at least $\frac{1}{2}$, she can win the bidding by matching \PT's budget, namely bidding $1-x^{(n)}_{\mathrm{wr},\mathrm{left}}$. Moreover, she must do so or lose the game. 
\PO then pays that amount to \PT, and moves to $v_\mathrm{right}$. Similarly, when in $v_\mathrm{right}$ and about to bid, \PO wins instantly if $x^{(n)}_{\mathrm{wr},\mathrm{right}}>\frac{1}{2}$, and otherwise \PT pays $x^{(n)}_{\mathrm{wr},\mathrm{right}}$ to \PO and moves to $v_\mathrm{left}$.

Starting with $\vec{w}^{(0)}_{\mathrm{left}}=\icol{\xinit\\1-\xinit}$, the sequence of vectors (while oscillating between $v_{\mathrm{left}}$ and $v_{\mathrm{right}}$) can therefore be described with the following steps (viewing $\WR$ as a matrix):
\begin{align*}
    &\text{(Step 1) } \vec{w}^{(n)}_{\mathrm{wr},\mathrm{left}}=\WR \vec{w}^{(n)}_{\mathrm{left}} \qquad \text{ (Step 2) } \vec{w}^{(n)}_{\mathrm{right}}=B_\mathrm{left}\vec{w}^{(n)}_{\mathrm{wr},\mathrm{left}} \\
    & \text{(Step 3) } \vec{w}^{(n)}_{\mathrm{wr},\mathrm{right}}=\WR \vec{w}^{(n)}_{\mathrm{right}} \qquad \text{(Step 4) } \vec{w}^{(n+1)}_{\mathrm{left}}=B_\mathrm{right}\vec{w}^{(n)}_{\mathrm{wr},\mathrm{right}}
\end{align*}
Where
\[
    B_\mathrm{left}=
    \begin{pmatrix}
        1  & -1 \\
        0 & 2
    \end{pmatrix}
    \qquad
    B_\mathrm{right}=
    \begin{pmatrix}
        2  & 0 \\
        -1 & 1
    \end{pmatrix}
    \qquad 
    \WR=
    \begin{pmatrix}
        1-\lambda & \lambda \\
        \lambda & 1-\lambda
    \end{pmatrix}
\]
Overall, $\vec{w}^{(n+1)}_{\mathrm{left}}=M\vec{w}^{(n)}_{\mathrm{left}}$ for the matrix $M=B_\mathrm{right}\WR B_\mathrm{left}\WR$, which depends on $\lambda$. 
The matrix $M$ has two eigenvalues: 
\begin{itemize}
    \item $\eValOne=1$, with the (normalized) eigenvector $\eVecOne=\icol{\frac{2\lambda-2}{4\lambda-3}\\ \frac{2\lambda-1}{4\lambda-3}}$.
    \item $\eValTwo=16\lambda^2-16\lambda+4$, with the eigenvector $\eVecTwo=\icol{-1\\1}$.
\end{itemize}

Recall that the budget vectors belong to the affine subspace $W=\{(x,y)\in \bbR^2 \mid x+y=1\}$, which is invariant under $M$. Every vector $\vec{w}\in W$ can be written as a linear combination of the form $\vec{w} = \eVecOne + c\eVecTwo$ for some $c\in\bbR$. Indeed, since $\eVecOne\in W$ and $\eVecTwo$ is the slope of $x+y=1$, we have that $W=\{\eVecOne+c\eVecTwo\mid c\in \bbR\}$.

We then have $M\vec{w} = \eVecOne + c\eValTwo\eVecTwo = \eVecOne + \eValTwo(\vec{w}-\eVecOne)$. Projecting this on the first coordinate and denoting $x_{\mathrm{fix}}=\frac{2\lambda-2}{4\lambda-3}$, for every $n$ we have $x^{(n+1)}_{\mathrm{left}} = \mathtt{next}\left(x^{(n)}_{\mathrm{left}}\right)$, where $\mathtt{next}(x) = x_{\mathrm{fix}} + \eValTwo(x-x_{\mathrm{fix}})$, and overall
\[
    x^{(n)}_{\mathrm{left}} = x_{\mathrm{fix}} + \left(\eValTwo\right)^n(\xinit-x_{\mathrm{fix}})
\]

If the condition $x^{(n)}_{\mathrm{wr},\mathrm{right}}>\frac{1}{2}$ is met for some $n$ (and the play does not end before that), \PO wins. Similarly, if $x^{(n)}_{\mathrm{wr},\mathrm{left}}<\frac{1}{2}$ (equivalently, $x^{(n)}_{\mathrm{left}}<\frac{1}{2}$) then \PO loses. If neither of these occurs for any $n$ then the play is infinite and \PO loses. 

It is convenient to phrase the win condition $x^{(n)}_{\mathrm{wr},\mathrm{right}}>\frac{1}{2}$ as a condition on $x^{(n)}_{\mathrm{left}}$, which allows the analysis to focus on $x^{(n)}_{\mathrm{left}}$ only. Note that $x^{(n)}_{\mathrm{wr},\mathrm{right}}$ is an injective function of $x^{(n)}_{\mathrm{left}}$, obtained by projecting the operator $\WR B_\mathrm{left} \WR$ on the first coordinate. Its reverse function, denoted $f^{\mathrm{rev}}(x)$, satisfies $x^{(n)}_{\mathrm{left}}=f^{\mathrm{rev}}\left(x^{(n)}_{\mathrm{wr},\mathrm{right}}\right)$. It is easy to verify that $f^{\mathrm{rev}}(x)=\frac{4\lambda^2-5\lambda+x+1}{2(2\lambda-1)^2}$, and it is increasing for all $x$. Therefore, \PO wins in the $n$'th iteration if and only if $x^{(n)}_{\mathrm{left}}>f^{\mathrm{rev}}(\frac{1}{2})$. It follows that \PO wins the play if and only if there exists $n$ such that $x^{(n)}_{\mathrm{left}}>f^{\mathrm{rev}}(\frac{1}{2})$ and $x^{\left(n'\right)}_{\mathrm{left}}\geq\frac{1}{2}$ for all $n'<n$.

We now split the analysis according to the value of $\lambda$. Specifically, according to whether $\eValTwo$ is the dominant eigenvalue, i.e., whether 
$\lambda<\frac{1}{4}$ or $\lambda\geq\frac{1}{4}$.

\paragraph{The case $\frac{1}{4}\leq\lambda<\frac{1}{2}$}

In this case, $f^{\mathrm{rev}}(\frac{1}{2})\geq1$ and therefore \PO's winning condition is never met, and she loses for every initial budget.

\paragraph{The case $0<\lambda<\frac{1}{4}$}
In this case, we have $\eValTwo>1$. Recall that $\eVecOne=\icol{x_{\mathrm{fix}} \\ 1-x_{\mathrm{fix}}}$ is an eigenvector with eigenvalue $1$, and we claim that it is the threshold vector, that is, $x_{\mathrm{fix}}=\frac{2\lambda-2}{4\lambda-3}$ is a (2-strong) threshold. Note that $x_{\mathrm{fix}}$ is increasing and continuous in $\lambda$, and equals $\frac{2}{3}$ for $\lambda=0$ and $\frac{3}{4}$ for $\lambda=\frac{1}{4}$.

Recall that $x^{(n)}_{\mathrm{left}} = x_{\mathrm{fix}} + \left(\eValTwo\right)^n(\xinit-x_{\mathrm{fix}})$, and \PO wins if and only if this value goes above $f^{\mathrm{rev}}(\frac{1}{2})$ (and does not go below $\frac{1}{2}$ before that). We now have $\frac{1}{2}<\frac{2}{3}\leq x_{\mathrm{fix}}<f^{\mathrm{rev}}\left(\frac{1}{2}\right)<1$. It follows that:
\begin{itemize}
    \item For $\xinit=x_{\mathrm{fix}}$, we have $x^{(n)}_{\mathrm{left}}\equiv x_{\mathrm{fix}}<f^{\mathrm{rev}}(\frac{1}{2})$ for all $n$, and so \PO loses.
    \item For $\xinit>x_{\mathrm{fix}}$, the sequence $x^{(n)}_{\mathrm{left}}$ increases unboundedly, eventually above $f^{\mathrm{rev}}(\frac{1}{2})$, at which point \PO wins.
    \item For $\xinit<x_{\mathrm{fix}}$, the sequence decreases unboundedly, eventually below $\frac{1}{2}$, at which point \PO loses.
\end{itemize}

\begin{remark}
    In addition to the thresholds, an analysis of the eigenvalues of $M$ gives us an insight regarding the behavior of the budget vectors throughout the play.
    
    In the case $\lambda=\frac{1}{4}$, $M$ is the identity matrix, meaning $x^{(n+1)}_{\mathrm{left}}=x^{(n)}_{\mathrm{left}}$ for all $n$, and the resulting sequence of configurations is periodic. We use the fact that a win condition for \PT is not met during the first iteration, since we only consider $\xinit\geq\frac{1}{2}$.

    In the case $\frac{1}{4}<\lambda<\frac{1}{2}$, $v_2$ cannot be reached symmetrically to $v_1$, and so the play is again infinite. In this case, $\eValTwo<1$. Since the dominant eigenvalue is $\eValOne=1$, the sequence of the budget vectors $v^{(n)}_{\mathrm{left}}$ converges to $\eVecOne$, which depends only on $\lambda$ and not on the initial budget.

     In the case $0<\lambda<\frac{1}{4}$, as mentioned above, the dominant eigenvalue is $\eValTwo>1$, and so the budgets would diverge if the play had stayed in $v_\mathrm{left},v_\mathrm{right}$.
\end{remark}

\subsection{Correct Tie Breaking}
\label{sec:example correct tie breaking}
We now return to the original tie breaking mechanism, where all ties are broken in favor of \PO. When bidding in $v_\mathrm{right}$, \PT must win the bidding in order to not lose instantly. The analysis must therefore be modified in the following ways:
\begin{itemize}
    \item \PO can win from $v_\mathrm{right}$ with a budget of exactly $\frac{1}{2}$. The winning condition is therefore changed to $x^{(n)}_{\mathrm{left}}\geq f^{\mathrm{rev}}(\frac{1}{2})$ (instead of a strict inequality). Accordingly, we differentiate between the cases $\lambda=\frac{1}{4}$ (where $f^{\mathrm{rev}}(\frac{1}{2})=1$) and $\lambda>\frac{1}{4}$ (where $f^{\mathrm{rev}}(\frac{1}{2})>1$).
    \item When in $v_\mathrm{right}$ and $x^{(n)}_{\mathrm{wr},\mathrm{right}}<\frac{1}{2}$, \PT must bid strictly more than $x^{(n)}_{\mathrm{wr},\mathrm{right}}$, i.e. $x^{(n)}_{\mathrm{wr},\mathrm{right}}+\epsilon_n$ for some $\epsilon_n>0$. Note that thus far, our analysis considered fixed optimal strategies, and hence a single play. Now, however, each strategy of \PT may choose different values for the $\epsilon_n$'s, thus inducing multiple plays that need to be analyzed.
\end{itemize}

\paragraph{The case $\lambda=\frac{1}{4}$}

In this case  $f^\mathrm{rev}(\frac{1}{2})=1$. With the updated winning condition, \PO wins if and only if $x^{(n)}_{\mathrm{left}}=1$ for some $n$.

For $\xinit<1$, \PO still loses. The reason is \PT can choose values $\epsilon_n$ so that the sequence $x^{(n)}_{\mathrm{left}}$ is increasing but strictly below $1$. 
Indeed, let $\{x_n\}_{n=0}^\infty$ be an increasing sequence, strictly bounded above by $1$, such that $x_0=\xinit$. We show that \PT can ensure $x^{(n)}_{\mathrm{left}}\leq x_n$ for all $n$. Fix $n$ and assume $x^{(n)}_{\mathrm{left}}\leq x_n$. The next value $x^{(n+1)}_{\mathrm{left}}$ is an increasing, continuous function of $\epsilon_n$. Since $M$ is the identity matrix, this function equals $x^{(n)}_{\mathrm{left}}$ for $\epsilon_n=0$, and so a small enough $\epsilon_n$ achieves  $x^{(n+1)}_{\mathrm{left}}\leq x_{n+1}$.

If $\xinit=1$, however, \PO wins in the first iteration. Indeed, by bidding $\frac{1}{4}$, moving to $v_\mathrm{right}$, then bidding $\frac{1}{2}$ and moving to $v_1$, the resulting play is:
\begin{align*}
    \underbrace{\icol{1\\0}}_{v_\mathrm{left}}
    \overset{\WR}{\longrightarrow}
    \underbrace{\icol{0.75\\0.25}}_{v_\mathrm{left}}
    \overset{\text{ bid }\frac14}{\longrightarrow}
    \underbrace{\icol{0.5\\0.5}}_{v_\mathrm{right}}
    \overset{\WR}{\longrightarrow}
    \underbrace{\icol{0.5\\0.5}}_{v_\mathrm{right}}
    \overset{\text{ bid }\frac12}{\longrightarrow}
    \underbrace{\icol{0\\1}}_{v_\mathrm{1}}
\end{align*}

\paragraph{The case $\frac{1}{4}<\lambda<\frac{1}{2}$}

Here $f^{\mathrm{rev}}(\frac{1}{2})>1$, and so \PO loses for any initial budget (including exactly $1$).

\paragraph{The case $0\leq\lambda<\frac{1}{4}$}
In this case $\eValTwo>1$, $x_{\mathrm{fix}}<f^{\mathrm{rev}}\left(\frac{1}{2}\right)<1$ and recall that in~\cref{sec:example alternative tie breaking} the threshold was $x_{\mathrm{fix}}$. We claim this remains the threshold, only it is now 1-strong instead of 2-strong. 

First consider $\xinit\geq x_{\mathrm{fix}}$. Fix a strategy of \PT. Since $\epsilon_0=0$ results in \PT losing in the first iteration, we assume for the rest of the analysis that $\epsilon_0>0$, and so $x^{(1)}_{\mathrm{left}}>\xinit\geq x_{\mathrm{fix}}$. Recall that in~\cref{sec:example alternative tie breaking}, such a value for the initial budget resulted in a winning play for \PO. Now, \PT can do no better; indeed, for all $n$, since $\epsilon_n>0$, we have by induction that $x_{\mathrm{fix}} + \left(\eValTwo\right)^{n-1}\left(x^{(1)}_{\mathrm{left}}-x_{\mathrm{fix}}\right)\leq x^{(n)}_{\mathrm{left}}$ (c.f., a similar equation but with equality in~\cref{sec:example alternative tie breaking}). We conclude that $x^{(n)}_{\mathrm{left}}$ increases above $f^{\mathrm{rev}} (\frac{1}{2})$, and \PO wins.

Finally, consider $\xinit<x_{\mathrm{fix}}$. This previously resulted in a sequence $x^{(n)}_{\mathrm{left}}$ with the (negative) differences $x^{(n)}_{\mathrm{left}}-x_{\mathrm{fix}}$ being multiplied by the constant $\eValTwo>1$ with each iteration, and the sequence $\epsilon_n$ can now be chosen to retain a similar behaviour of the differences. Indeed, Fixing $n$ and $x^{(n)}_{\mathrm{left}}$, we have that $x^{(n+1)}_{\mathrm{left}}$ is an increasing, continuous function of $\epsilon_n$, that equals $x_{\mathrm{fix}}+\eValTwo\left(x^{(n)}_{\mathrm{left}}-x_{\mathrm{fix}}\right)$ for $\epsilon_n=0$. A small enough $\epsilon_n$ can therefore achieve $x^{(n+1)}_{\mathrm{left}} < x_{\mathrm{fix}}+\frac{1}{2}\eValTwo\left(x^{(n)}_{\mathrm{left}}-x_{\mathrm{fix}}\right)$. In conclusion, \PT can still force the sequence $x^{(n)}_{\mathrm{left}}$ to decrease below $\frac{1}{2}$, and so \PO loses.

Concluding the three cases, we have that the thresholds are the following:
\begin{itemize}
    \item For $0\leq\lambda<\frac{1}{4}$: $\tau(\lambda)=\frac{2\lambda-2}{4\lambda-3}$ (increasing from $\frac{2}{3}$ to $\frac{3}{4}$, 1-strong threshold)
    \item For $\lambda=\frac{1}{4}$: $\tau(\lambda)=1$ (1-strong threshold)
    \item For $\frac{1}{4}<\lambda<\frac{1}{2}$: $\tau(\lambda)=1$ (2-strong threshold)
\end{itemize}
Note that $\tau(\lambda)$ is discontinuous at $\lambda=\frac{1}{4}$, which gives us the following.
\begin{corollary}
    \label{cor:threshold not continuous}
    There exists a game $\cG$ and vertex $v$ such that the threshold function of $v$ is discontinuous as a function of $\lambda$.
\end{corollary}
\begin{remark}
The behavior of $\tau$ in this example can be given an economic interpretation: after a certain threshold (namely $\lambda=\frac14$), the threshold (suddenly) becomes equivalent to that of any $\frac14\le \lambda<\frac12$. This suggests that beyond a certain threshold, it no longer helps anyone to impose more tax. Naturally this does not fully extend to real-life economics, but it is a curiosity nonetheless.
\end{remark}
\begin{remark}
    \label{rmk:toolbox}
    The analysis carried out in this section crucially depends on obtaining a characterization of the play resulting from optimal strategies as a linear dynamical system. Since analyzing such systems is notoriously difficult, especially in high dimensions~\cite{ouaknine2012decision}, automating this analysis algorithmically seems out of reach. 
    Nonetheless, the tools we develop in this section may be used in other examples. Specifically, starting with alternative tie-breaking to avoid sinks, and using the dominant eigenvalues as a guide to the long-run behavior.
\end{remark}

\section{Existence of Thresholds}
\label{sec: thresholds main}
The analysis in~\cref{sec: example} demonstrates the threshold function for a specific game, but does not give a general technique for computing thresholds, nor shows that they always exist.
In this section we present our main result, namely that every game has a threshold function.

The section is organized as follows. In~\cref{sec:average property} we describe an invariant dubbed \emph{the average property} 
which gives a lower bound on the threshold.
In~\cref{sec:threshold on DAGs} we restrict the discussion to games played on directed acyclic graphs (DAGs), in which case there exists a unique function satisfying the average property, and it constitutes a threshold. 
In~\cref{sec:threshold on General Graphs} we turn to general graphs, and show the existence of a threshold by a reduction to the setting of DAGs. In~\cref{sec:threshold in terms of average property} we show that the threshold satisfies the average property, and is obtained as the point-wise maximum over all functions satisfying this property. Finally, in~\cref{sec:threshold complexity} we use this characterization to compute the threshold using mixed-integer linear programming (MILP).

\subsection{The Average Property}
\label{sec:average property}
Consider a game $\cG=\tup{G,\lambda,T}$, and assume it has a threshold $\Th\colon{V\to[0,1]}$. For a sink $v\in T$, it holds that $\Th(v)=0$, since starting at the target, \PO instantly wins for any initial budget. For a sink $v\notin T$, we have $\Th(v)=1$, since \PO instantly loses for any initial budget.

For a non-sink $v\in V$, $\Th(v)$ relates to the minimum and maximum values of $\Th$ among $v$'s neighbors as follows. Intuitively, if \PO wins the bidding at $v$, it is optimal for her to choose the next vertex to have a minimal threshold. Similarly, \PT would choose the maximal threshold. As we show in this section, the budget needed \emph{during the bidding phase} in $v$ for \PO to win the game turns out to be exactly the average of these two values, and so the threshold (which is the budget before WR) is obtained by applying the reverse map $\WRinv$ on this average (where $\WRinv(x)=\frac{x-\lambda}{1-2\lambda}$). 
We remark that similar average properties typically arise in bidding games~\cite{avni2020CONCUR}.
\begin{definition}
\label{def: average property}
    Let $G=(V,E)$ be a graph, $T\subseteq V$ a subset of the sinks, and $f\colon V\to[0,1]$. For a non-sink $v\in V$, let
    \[
        v^+ = \argmax_{u\in \neig(v)} f(u) \quad\text{ and }
        \quad v^- = \argmin_{u\in \neig(v)} f(u)
    \]
    (we choose arbitrarily if the extrema are not unique). 
    Let 
    \[
        f_{\mathrm{avg}}(v)=\frac{f(v^+)+f(v^-)}{2}    \quad\text{ and }
        \quad\ f_\pre(v)=\WRinv(f_\mathrm{avg}(v))\
    \]
    We say that $f$ satisfies the \emph{average property} if for every $v\in V$:
    \begin{itemize}
        \item If $v$ is a sink, then
        $
            f(v)= \begin{cases}
                        0 & v\in T \\
                        1 & v\notin T
                    \end{cases}
        $.
        \item If $v$ is not a sink, then 
        $
            f(v)= \begin{cases}
                        0 & f_\pre(v) < 0 \\
                        f_\pre(v) & f_\pre(v)\in[0,1] \\
                        1 & f_\pre(v) > 1 \\
                    \end{cases}
        $
        Equivalently, $f(v)=\max(\min(f_\pre(v),1),0)$
    \end{itemize}
\end{definition}
Note that the ``artificial'' introduction of $1$ and $0$ as limits makes sense both because $f$'s range is $[0,1]$, but also semantically. For example, if $f_\pre(v)<0$, it can be viewed as ``\PO does not even need $0$ in order to win, so $0$ is certainly enough''. Formally, since $\WR(x)$ is an increasing function, we have in this case that $\WR(0)>\WR(f_\pre(v))=f_\mathrm{avg}(v)$. Therefore, starting with budget $0$, during the first bidding \PO has more than $f_\mathrm{avg}(v)$. In the following we show  this is enough for \PO to win. The analysis for $f_\pre(v)>1$ is similar.

The following lemma provides a clear view of the motivation for the average property and its relation to thresholds. Intuitively, it states that if $f$ satisfies the average property, then starting from $\xinit>f(v_0)$, \PO can guarantee that the current budget always remains above $f(v)$, and dually for \PT staying below $f(v_0)$.
At first glance, this may seem to suggest that every function satisfying the average property is a threshold. This, however, is generally false: there may be multiple such functions, while the threshold is clearly unique.
\begin{lemma}
\label{lem: strategies for average property}
    Let $\cG=\tup{G,v_0,\lambda,T}$ be a game, and let $f$ be a function satisfying the average property. 
    There exist strategies $\sigma_1,\sigma_2$ for Players 1 and 2, respectively, such that the following holds.
    \begin{enumerate}
        \item If $\xinit>f(v_0)$ then for every strategy $\sigma'_2$ of \PT, every configuration $(v,x)$ in $\play(\sigma_1,\sigma'_2,v_0,\xinit)$ satisfies $x>f(v)$.
        \item 
        If $\xinit<f(v_0)$ then for every strategy $\sigma'_1$ of \PO, every configuration $(v,x)$ in $\play(\sigma'_1,\sigma_2,v_0,\xinit)$ satisfies $x<f(v)$.
    \end{enumerate}
\end{lemma}

\begin{proof}
Assume $\xinit>f(v_0)$. We describe $\sigma_1$ inductively. Let $(v,x)$ be a configuration such that $v$ is not a sink and $x>f(v)$. Then $f(v)<1$, and in particular $f(v)\geq f_\pre(v)$ (as either $f(v)=f_\pre(v)$ or $f(v)=0>f_\pre(v)$). After WR, \PO has budget 
\[
    \WR(x) >  \WR(f(v)) 
    \geq  \WR(f_\pre(v)) 
    =  f_\mathrm{avg}(v)
\]
We now describe the bid of \PO. Let
\[
    f_\mathrm{diff}(v)=\frac{f(v^+)-f(v^-)}{2}
\]
\PO bids $f_\mathrm{diff}(v)$ (note that $f_\mathrm{diff}(v)\le f_\mathrm{avg}(v)$). If she wins the bidding, she moves to $v^-$, at which point her budget is $x'$ satisfying 
\[
    x' =\WR(x)-f_\mathrm{diff}(v)> f_\mathrm{avg}(v)-f_\mathrm{diff}(v)=f(v^-)
\]
As desired. 
Dually, if \PO loses the bidding, then \PT bid more than $f_\mathrm{diff}(v)$, so \PO's new budget is $x'$ satisfying
\[
    x' > \WR(x) + f_\mathrm{diff}(v) > f_\mathrm{avg}(v)+f_\mathrm{diff}(v)=f(v^+)
\]
and the invariant is maintained regardless of the vertex \PT chooses to move to (since $v^+$ has maximal value of $f$ among the neighbors).

Next assume $\xinit<f(v_0)$, we describe $\sigma_2$. Given $(v,x)$ such that $x<f(v)$, we have in particular $f(v)>0$, and so $f(v)\leq f_\pre(v)$. After WR, \PO has budget
$
    \WR(x) <  \WR(f(v)) 
    \leq  \WR(f_\pre(v)) 
    =  f_\mathrm{avg}(v) 
    =  f(v^+)-f_\mathrm{diff}(v) 
    \leq  1-f_\mathrm{diff}(v)
$.
Thus, \PT has budget at least $f_\mathrm{diff}(v)$. He bids that amount, and upon winning moves to $v^+$. 
If \PO wins the bidding, her new budget is $x'< f_\mathrm{avg}(v)-f_\mathrm{diff}(v)=f(v^-)$ and so the invariant is maintained in the next vertex regardless of \PO's choice (since $v^-$ has minimal value of $f$ among the neighbors). 
If \PO loses the bidding, she has less then $f(v^+)$, as desired.
\end{proof}
Consider a function $f$ satisfying the average property, and recall that \PO wins in a play if and only if a vertex $v\in T$ is reached. Such vertices satisfy $f(v)=0$ by~\cref{def: average property}.
Thus, if every configuration $(v,x)$ in a play satisfies the invariant $x<f(v)$, then it cannot hold that $f(v)=0$ for any vertex in that play, i.e. the play does not reach $T$. Using~\cref{lem: strategies for average property} we then have the following.
\begin{corollary}
\label{cor: average property <= threshold}
    Let $\cG$ be a reachability game, and let $f$ be a function satisfying the average property. If $\xinit<f(v_0)$, then \PT has a winning strategy.
\end{corollary}

Note that a dual argument for \PO winning when $\xinit>f(v_0)$ fails, as her losing does not require reaching $f(v)=1$.

\subsection{Games Played on Directed Acyclic Graphs}
\label{sec:threshold on DAGs}
In this section we restrict attention to directed acyclic graphs (DAGs), and show the existence of thresholds in this case. We rely on these results in~\cref{sec:threshold on General Graphs} where we generalize to all graphs.
\begin{lemma}
\label{lem: existence of threshold for reachability DAGs}
    Consider a game $\cG=\tup{G,\lambda,T}$ such that $G$ is a DAG, then $\cG$ has a unique function that satisfies the average property, and it is a threshold function.
\end{lemma}
\begin{proof}
We start by constructing a function $f$ that satisfies the average property. This is done bottom-up from the sinks (i.e., leaves) of $G$. Recall that for every sink $v$, if $v\in T$ then $f(v)=0$ and otherwise $f(v)=1$. In particular, note that any function that satisfies the average property must coincide with $f$ on the sinks.
We proceed by defining $f$ inductively: consider a vertex $v$ such that $f(u)$ is defined for every $u\in \neig(v)$, then we can compute $v^+$ and $v^-$, and proceed to define $f(v)$ as per \cref{def: average property}. Note that since $G$ is acyclic, this process terminates. Moreover, $f$ is defined uniquely for every vertex using this process.
Denote the unique function above as $\Th$, we prove that it is indeed a threshold function. Consider a configuration $(v_0,\xinit)$.

If $\xinit<\Th(v_0)$ then \PT wins by \cref{cor: average property <= threshold}.

If $\xinit>\Th(v_0)$ then since $G$ is a DAG, a dual argument to \cref{cor: average property <= threshold} applies: every play ends in a sink, and so \PO losing is equivalent to the play reaching a vertex $v$ with $\Th(v)=1$. The strategy given by \cref{lem: strategies for average property} maintains $\Th(v)<x\leq 1$, thus winning for \PO.
\end{proof}

\begin{example}
We illustrate \cref{lem: existence of threshold for reachability DAGs} in \cref{fig: example DAG} with $\lambda=\frac16$ (the names of the vertices are relevant for \cref{sec:threshold on General Graphs}). Observe that for $\lambda=\frac16$ we have $\WRinv(x)=\frac{3}{2}x-\frac14$. Thus, in $(v_1,1)$ we have $\Th_{\mathrm{avg}}(v_1,1)=\frac12$, so $\Th(v_1,1)=\WRinv(\frac12)=\frac12$. 
In $(v_0,1)$ we have $\Th_{\mathrm{avg}}(v_0,1)=1$, so $\Th_{\mathrm{pre}}=\WRinv(1)=\frac54>1$, and therefore $\Th(v_0,1)=1$.
Finally, $\Th_{\mathrm{avg}}(v_0,0)=\frac34$ so $\Th(v_0,0)=\WRinv(\frac34)=\frac78$.
\end{example}

\begin{figure}[ht]
\centering  
\begin{tikzpicture}[auto,node distance=2.8cm,scale=1]
    \tikzset{every state/.style={ellipse, draw, minimum size=15pt, inner sep=0}};
    \node (q00) [state] at (0,0) {$(v_0,0),\frac{7}{8}$};
    \node (q01) [state] at (3,0.5) {$(v_0,1),1$};
    \node (q11) [state] at (3,-0.5) {$(v_1,1),\frac{1}{2}$};
    \node (q02) [state] at (6,1) {$(v_0,2),1$};
    \node (q12) [state] at (6,0.34) {$(v_1,2),1$};
    \node (q22) [state] at (6,-0.33) {$(v_2,2),1$};
    \node (q32) [state, accepting] at (6,-1) {$(v_3,2),0$};
    \path [-stealth]
    (q00) edge [bend left=0] (q11)
    (q11) edge [bend right=0] (q22) 
    (q11) edge [bend left=0] (q32)
    (q00) edge [bend right=0] (q01)
    (q01) edge [bend left=0] (q12)
    (q01) edge [bend right=0] (q02)
    ;
\end{tikzpicture}
\caption{A DAG (in this case, a tree) and the unique values satisfying the average property, for $n=2$ and $\lambda=\frac{1}{6}$.}
\label{fig: example DAG}
\Description{A tree of $7$ vertices}
\end{figure}
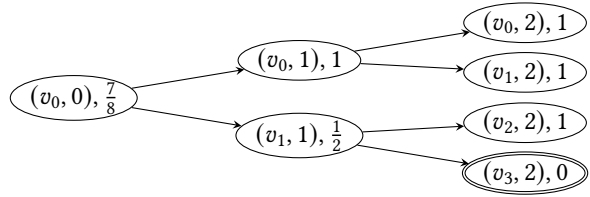

We now turn our attention to what happens when the initial budget equals exactly the threshold.

\begin{proposition}
\label{claim: strength of threshold for reachability DAGs}
    Let $\cG$ be a game played on a DAG. Then $\Th$ is either 1-strong or 2-strong, and satisfies the following.
    \begin{itemize}
        \item For a sink $v$, it is 1-strong if and only if $v\in T$.
        \item For a non-sink $v$ and $\lambda=0$, it is 1-strong if and only if $T$ is reachable from $v$ (equivalently, $\Th(v)<1$).
        \item For a non-sink $v$ and $\lambda>0$, it is 1-strong if and only if $\Th_\pre(v)\leq1$.
    \end{itemize}
\end{proposition}

\begin{proof}
We consider the different cases mentioned.
\paragraph{The case that $v$ is a sink}
The thresholds $0$ and $1$ for $v\in T$ and $v\notin T$ respectively, are trivially 1-strong and 2-strong.

\paragraph{The case that $v$ is not a sink and $\lambda=0$}
It is shown in~\cite{avni2021bidding} that $T$ is reachable from $v$ if and only if $\Th(v)<1$. If it is not, then the threshold is clearly 2-strong. Assume $T$ is reachable from $v$, and we show that $\Th(v)$ is 1-strong, by induction from the sinks backwards. 

We claim that for every $u\in \neig(v)$, we have that $T$ is reachable from $u$ if and only if $\Th(u)$ is 1-strong. Indeed, if $u$ is a sink then $T$ is reachable from $u$ if and only if $u\in T$ if and only if $\Th(u)$ is 1-strong; otherwise, the equivalence follows from the induction hypothesis. Note that, since $v^-$ is defined to be chosen arbitrarily among the neighbors of $v$ with the lowest thresholds, we can choose it to have a 1-strong or a weak threshold when possible (intuitively, we consider a 2-strong thresholds ``higher'' than other thresholds of the same value). Under this convention, assume by way of contradiction that $T$ is not reachable from $v^-$. Then $\Th(v^-)=1$ is 2-strong and so $\Th(u)=1$ is a 2-strong threshold for all $u\in \neig(v)$. In this case, $T$ is not reachable from any such $u$, in contradiction to it being reachable from $v$. Therefore, $T$ is reachable from $v^-$, and so $\Th(v^-)$ is 1-strong. It remains to show that \PO has a winning strategy from $(v,\Th(v))$. Indeed, the strategy of bidding $\Th_\mathrm{diff}(v)$ wins similarly to before: Upon winning the bidding she has exactly $\Th(v^-)$, which wins from $v^-$. Upon losing the bidding she has strictly more than $\Th(v^+)$ and so she wins from any possible next vertex.

\paragraph{The case that $v$ is not a sink and $\lambda>0$}
We show that $\Th(v)$ is 1-strong if and only if $\Th_\pre(v)\leq 1$.

Assume $\Th_\pre(v)\leq 1$, and consider the configuration $(v,\Th(v))$. \PO's budget during the first bidding, namely $\WR(\Th(v))$, is at least $\Th_\mathrm{avg}(v)$, and she wins by bidding $\Th_\mathrm{diff}(v)$ similarly to before, provided that $\Th(v^-)$ is 1-strong. Indeed, $\Th(v^-)\leq \Th_{\mathrm{avg}}(v)\ = \WR(\Th_\pre(v)) < 1$, the strict inequality holding due to $\lambda>0$. If $v^-$ is a sink, it follows that $\Th(v^-)=0$, and it is 1-strong. Otherwise,  $\Th_\pre(v^-)\leq1$ and so $\Th(v^-)$ is 1-strong by the induction hypothesis.

If $\Th_\pre(v)>1$, then $\Th(v)<\Th_\pre(v)$. During the first bid \PO has $\WR(\Th(v))<\WR(\Th_\pre(v))=\Th_\mathrm{avg}(v)$, and \PT wins as before.
\end{proof}

Note that if $\Th(v)<1$ then $\Th_\pre(v)\leq 1$, so we have the following.
\begin{corollary}
\label{cor: threshold<1 is 1-strong for DAG}
    Let $\cG$ be a reachability game played on a DAG. Whenever $\Th(v)<1$, it is 1-strong.
\end{corollary}

\subsection{Games Played on General Graphs}
\label{sec:threshold on General Graphs}
We are now ready for our main result.
\begin{theorem}
\label{thm: existence of threshold for reachability games}
    Every game $\cG=\tup{G,\lambda,T}$ has a threshold function which satisfies the average property.
\end{theorem}
\begin{proof}
We construct a function $\Th(v)$, show that it satisfies the average property, and finally show that it constitutes a threshold function.

The first step is to reduce the setting to that of DAGs. Denote $G=(V,E)$ and let $n\in \bbN$. We turn to define a game that is played on a DAG (specifically, the unravelling of $G$ for $n$ steps) and, intuitively, corresponds to the objective of winning in $\cG$ within at most $n$ steps. 
Consider the DAG $\unravel_n=(V\times \{0,\ldots,n\},E_n)$, where $E_n=\{(u_i,v_{i+1})\mid(u,v)\in E,0\leq i\leq n-1\}$. As an example, the underlying graph of the game depicted in \cref{fig: average property non uniqueness} yields the DAG depicted in \cref{fig: example DAG} for $n=2$.

Next, we define the game $\cG_n=\tup{\unravel_n,\lambda,T\times \{0,\ldots,n\}}$. By \cref{lem: existence of threshold for reachability DAGs} we have that $\cG_n$ has a threshold function $\Th_n$. 

Let $v\in V$. We consider the relation between $\Th_n((v,0))$ and $\Th_{n+1}((v,0))$. Assume \PO has a winning strategy $\sigma_1$ for $\cG_n$, starting in configuration $((v,0),x_\mathrm{init})$. Following $\sigma_1$ ensures, in particular, that the play does not reach $(V\times\{n\})\setminus(T\times\{n\})$, as those are sinks not belonging to the target. 
Observe that $\unravel_{n+1}$ is obtained from $\unravel_n$ by (possibly) adding outgoing edges only from $(V\times\{n\})\setminus(T\times\{n\})$.
The strategy $\sigma_1$ therefore wins, starting in $((v,0),x_\mathrm{init})$, in $\cG_{n+1}$ as well. Intuitively, winning $\cG$ in at most $n$ steps particularly wins it in at most $n+1$ steps. Thus, $\Th_{n}((v,0))\ge \Th_{n+1}((v,0))$, i.e., the sequence $\{\Th_n((v,0))\}_{n=0}^\infty$ is non-increasing. This sequence is also bounded from below by $0$, and therefore converges. We  define the threshold $\Th$ for $\cG$ as the pointwise-limit 
\[
    \Th(v)=\lim_{n\to\infty}\Th_n((v,0))
\]

Next, we prove that $\Th$ satisfies the average property.
For a sink $v\in T$, we have $\Th_n((v,0))=0$ for all $n$ (since $(v,0)\in T\times\{0,\ldots,n\}$), and so $\Th(v)=0$ as needed. For a sink $v\notin T$, we have $\Th_n(v)=1$ for all $n$ (since $(v,0)$ is a sink in $\unravel_n$ and does not belong to the target), and so $\Th(v)=1$ as needed. 

For the following, fix a non-sink $v$. We need to show (as per~\cref{def: average property}) that 
\begin{equation}
\label{eq: average property for Th}
    \Th(v) = \max\left(\min\left(\WRinv\left(\frac{\Th(v^+)+\Th(v^-)}{2}\right),1\right),0\right)
\end{equation}
Note that $(v,0)$ is not a sink in $\unravel_n$ for all $n\geq 1$. For every $n\geq 1$, let 
\[
    v^+_n = \argmax_{u\in \neig(v)} \Th_n(u) \quad \text{ and }\quad v^-_n  = \argmin_{u\in \neig(v)} \Th_n(u)
\]
For every $u\in \neig(v)$, consider the sub-DAG of $\unravel_n$ starting in $(u,1)$. Observe that this sub-DAG is isomorphic to the sub-DAG of $\unravel_{n-1}$ starting in $(u,0)$ (with the difference only being the indices of the levels). It follows that $\Th_n((u,1))=\Th_{n-1}((u,0))$. By the average property for $\cG_n$, we have 
\[
\Th_n((v,0))=\max\left(\min\left(\WRinv\left(\frac{\Th_{n-1}((v^+_n,0))+\Th_{n-1}((v^-_n,0))}{2}\right),1\right),0\right)
\]
Note that this is a continuous function of $\Th_{n-1}((v^+_n,0))$ and $\Th_{n-1}((v^-_n,0))$. 
In order to show \cref{eq: average property for Th}, it is therefore enough to show 
\begin{align}
    \label{eq: limit of max threshold} \lim_{n\to\infty}\Th_{n-1}((v^+_n,0))=\Th(v^+) \\
    \label{eq: limit of min threshold} \lim_{n\to\infty}\Th_{n-1}((v^-_n,0))=\Th(v^-) 
\end{align}
We show that \cref{eq: limit of max threshold} holds, and the proof for \cref{eq: limit of min threshold} is analogous. 
Note that $v^+_n$ might be a different vertex for each $n$, and so the left hand side of \cref{eq: limit of max threshold} does not describe the limit of thresholds for a single vertex. 
Intuitively, however, there is a set of vertices that appear infinitely often in this limit whose corresponding limits are all equal, which enables us to conclude the claim.

Formally, let $N^+(v)=\{u\in \neig(v) \mid \Th(u)=\Th(v^+)\}$. Then for every $u^+\in N^+(v)$ and $u'\in \neig(v)\setminus N^+(v)$, we have that $\Th(u')<\Th(u^+)$. By the definition of $\Th$ as a limit, it follows that there exists $n_{u^+,u'}\in \bbN$ such that $\Th_n((u',0))<\Th_n((u^+,0))$ for every $n>n_{u^+,u^-}$. 
Taking $n_0$ as the maximum over all such $n_{u^+,u^-}$, it follows that $v^+_n\in N^+(v)$ for every $n>n_0$. Since $N^+(v)$ is finite, and $\Th_{n-1}((u,0))\overset{n\to\infty}{\longrightarrow}\Th(v^+)$ for every $u\in N^+(v)$, 
we have that $\Th_{n-1}((v^+_n,0))\overset{n\to\infty}{\longrightarrow}\Th(v^+)$ as well.

Finally, we show that $\Th$ is a threshold function for $\cG$. 
Assume $\cG$ starts in $(v,x_\mathrm{init})$.

If $\xinit>\Th(v)$ then there exists $n$ such that $\xinit>\Th_n(v)$, meaning \PO has a strategy that wins in at most $n$ steps, and in particular wins in $\cG$.

If $\xinit<\Th(v)$, then since $\Th$ satisfies the threshold property, it follows from \cref{cor: average property <= threshold} that \PT has a winning strategy.
\end{proof}

\subsection{Characterization of Thresholds in Terms of the Average Property}
\label{sec:threshold in terms of average property}
Recall that for games played on DAGs, the threshold is the unique function satisfying the average property, and it can be computed inductively from the sinks. For general games with $\lambda=0$, it is known~\cite{avni2021bidding} that there is still a unique function satisfying the average property, and that finding the threshold is in $\mathsf{NP}\cap\mathsf{coNP}$ (and in $\mathsf{P}$ for graphs with out-degree $2$). 

In stark contrast, this uniqueness no longer holds in Robin Hood games for $\lambda>0$, as we now demonstrate.
Consider the game in \cref{fig: average property non uniqueness}, with $\lambda=\frac{1}{4}$. It can be verified that the numbers on the vertices are the thresholds. However, for any $t\in[0,1]$, the function defined by $f(v_0)=t$ and $f(v)=\Th(v)$ for $v\neq v_0$ satisfies the average property. Indeed, this can be easily checked for $v\neq v_0$ (and also follows from~\cref{lem: existence of threshold for reachability DAGs}, since from without $v_0$ the game is a DAG). 
For $v_0$, note that $\WRinv(x)=2(x-\frac{1}{4})=2x-\frac12$, so 
\[
    f_\pre(v_0) =  \WRinv(f_\mathrm{avg}(v_0)) 
    =  \WRinv\left(\frac{t+\frac{1}{2}}{2}\right) 
    =  t 
    =  f(v_0)
\]
and since $f(v_0)\in[0,1]$, the average property holds.

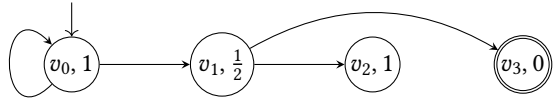
\begin{figure}[ht]
\centering  
\begin{tikzpicture}[auto,node distance=2.8cm,scale=1]
    \tikzset{every state/.style={minimum size=20pt, inner sep=1},};
    \node (q0) [initial above, state, initial text = {}] at (0,0) {$v_0,1$};
    \node (q1) [state] at (2,0) {$v_1,\frac{1}{2}$};
    \node (q2) [state] at (4,0) {$v_2,1$};
    \node (q3) [state, accepting] at (6,0) {$v_3,0$};
    \path [-stealth]
    (q0) edge [loop, distance=30,out=225,in=135] node {} (q0)
    (q0) edge (q1)
    (q1) edge (q2)
    (q1) edge [bend left] (q3);
\end{tikzpicture}
\caption{Infinitely many functions satisfying the average property, for $\lambda=\frac{1}{4}$. The numbers are the thresholds.}
\label{fig: average property non uniqueness}
\Description{A graph consisting of an initial vertex $v_0$ with value $1$ with a self loop and an edge going to $v_1$. $v_1$ has value $\frac{1}{2}$, and edges to $v_2$ and $v_3$ with values $1,0$ respectively.}
\end{figure}

As mentioned above, the threshold in~\cref{fig: average property non uniqueness} satisfies $\Th(v_0)=1$, which coincides with the maximal value of $t$. As it turns out, this is not a coincidence.
\begin{theorem}
\label{thm: threshold is max avg property}
    Consider a game $\cG$. The threshold $\Th$ is the point-wise maximum over the functions satisfying the average property.
\end{theorem}
\begin{proof}
Let $A$ be the set of functions satisfying the average property and let $m(v)=\sup_{f\in A}f(v)$. Since $\Th$ satisfies the average property, we have $\Th(v)\leq m(v)$. We additionally have $f(v)\leq\Th(v)$ for every $f\in A$; indeed, assume by way of contradiction that $\Th(v)<f(v)$, and let $\Th(v)<x_\mathrm{init}<f(v)$. By the threshold definition, \PO has a winning strategy, but by \cref{cor: average property <= threshold}, \PT also does. Therefore, $m(v)\leq\Th(v)$, and overall $\Th(v)=m(v)$.
\end{proof}

\subsection{MILP for Computing the Threshold}
\label{sec:threshold complexity}
While \cref{thm: existence of threshold for reachability games}
shows the existence of thresholds, it uses a limit and is therefore not constructive. However, using \cref{thm: threshold is max avg property} we can algorithmically compute the threshold using mixed-integer linear programming (MILP).
\begin{theorem}
\label{thm:milp for threshold}
    Given a game $\cG$, we can efficiently construct a MILP instance $\cI$ whose solution is the threshold function for $\cG$.
\end{theorem}
\begin{proof}
Consider a game $\cG$ with vertices $V=v_1,\ldots, v_m$. 
The average property can be readily expressed as a set of constraints on the variables $\Th(v)$ for every $v\in V$ containing linear and $\min$/$\max$ expressions (as per~\cref{def: average property}). 
Observe that this complies with MILP, since e.g., the expression $\min\{x_1,x_2\}$ can be removed by introducing a new variable $X$, a variable $b\in \{0,1\}$ and the following constraints:
\[
     X\leq x_1 \wedge 
     X\leq x_2\wedge 
     X\geq x_1-2M(1-b)\wedge 
     X\geq x_2-2Mb
\]
Where $M$ is a bound such that $|x_1|,|x_2|<M$. In our case, $M=\WRinv(1)$ is such a bound.
Indeed, the first two constraints ensure $X\le \min\{x_1,x_2\}$, and the latter two ensure that either $X\ge x_1$ (if $b=1$) or $X\ge x_2$ (if $b=0$). The choice of $M$ ensures that the two latter constraints are satisfiable.

Finally, we maximize the objective $\sum_{v\in V}\Th(v)$. The solution then equals the threshold by \cref{thm: threshold is max avg property}.
\end{proof}

\begin{example}
\label{xmp:milp}
     We demonstrate the construction of the MILP for the game depicted in \cref{fig: average property non uniqueness}. For each $v_i$, we use a variable $v_i$ to represent the value $f(v_i)$. The resulting MILP is in \cref{tab:milp}.

\begin{table}[h]
    \centering
    \begin{tabular}{|c|l|}
    \hline
    \multicolumn{2}{|c|}{maximize $v_1+v_2+v_3+v_4$ subject to:}\\
    \hline
    \hline
    \label{eq: MILP_table_sinks} (C1) & $v_2 = 1 \wedge v_3 = 0$ \\ \hline
    \label{eq: MILP_table_v^-_1} (C2) &
    \begin{tabular}[c]{@{}l@{}}
    $v^-_1 \leq v_2 \wedge v^-_1 \leq v_3$\\
    $v^-_1 \geq v_2 - 2M(1 - b^-_1) \wedge v^-_1 \geq v_3 - 2M b^-_1$
    \end{tabular} \\ \hline
    \label{eq: MILP_table_v^+_1} (C3) &
    \begin{tabular}[c]{@{}l@{}}
    $-v^+_1 \leq -v_2 \wedge -v^+_1 \leq -v_3$ \\
    $-v^+_1 \geq -v_2 + 2M(1 - b^+_1) \wedge -v^+_1 \geq -v_3 + 2M b^+_1$
    \end{tabular} \\ \hline
    \label{eq: MILP_table_v_1_min_with_1} (C4) &
    \begin{tabular}[c]{@{}l@{}}
    $v'_1 \leq \WR^{-1}\left(\frac{v^+_1 + v^-_1}{2}\right) \wedge v'_1 \leq 1$ \\
    $v'_1 \geq \WR^{-1}\left(\frac{v^+_1 + v^-_1}{2}\right) - 2M(1 - b'_1) \wedge v'_1 \geq 1 - 2M b'_1$
    \end{tabular} \\ \hline
    \label{eq: MILP_table_v_1_max_with_0} (C5) &
    \begin{tabular}[c]{@{}l@{}}
    $-v_1 \leq -v'_1 \wedge -v_1 \leq 0$ \\
    $-v_1 \geq -v'_1 + 2M(1 - b_1) \wedge -v_1 \geq 2M b_1$
    \end{tabular} \\ \hline
    \label{eq: MILP table v0} (C6) &
    \begin{tabular}[c]{@{}l@{}}
    $v^-_0 \leq v_0 \wedge v^-_0 \leq v_1$ \\
    $v^-_0 \geq v_0 - 2M(1 - b^-_0) \wedge v^-_0 \geq v_1 - 2M b^-_0$ \\
    $-v^+_0 \leq -v_0 \wedge -v^+_0 \leq -v_1$ \\
    $-v^+_0 \geq -v_0 + 2M(1 - b^+_0) \wedge -v^+_0 \geq -v_1 + 2M b^+_0$ \\
    $v'_0 \leq \WR^{-1}\left(\frac{v^+_0 + v^-_0}{2}\right) \wedge v'_0 \leq 1$ \\
    $v'_0 \geq \WR^{-1}\left(\frac{v^+_0 + v^-_0}{2}\right) - 2M(1 - b'_0) \wedge v'_0 \geq 1 - 2M b'_0$ \\
    $-v_0 \leq -v'_0 \wedge -v_0 \leq 0$ \\
    $-v_0 \geq -v'_0 + 2M(1 - b_0) \wedge -v_0 \geq 2M b_0$
    \end{tabular} \\ \hline
    (C7) & $b_1^-,b_1^+,b_1',b_1,b_0^-,b_0^+,b_0',b_0\in \{0,1\}$\\
    \hline
    \end{tabular}
    \caption{The MILP for \cref{xmp:milp}}
    \label{tab:milp}
\end{table}

     (C1) expresses the average property requirement for the sinks.
     For the non-sink $v_1$, the requirement of the average property
    involves $v^-_1$, which can attain the value of either $v_2$ or $v_3$, and therefore introduces a new variable\footnote{Observe that for games with out-degree at most 2, we can exploit the symmetry between $v^-_i$ and $v^+_i$, in that we do not need to encode which is the minimal and which is the maximal.} and the constraints in 
    (C2).
    The binary variable $b^-_1$ gets the value $1$ if $v_2=\min(v_2,v_3$), and $0$ otherwise. 
    The constraint (C3) similarly set $v^+_1=-\min(-v_2,-v_3)$.
%
    Next, 
    (C4)
    serves to express $v_1=\min\left(\WR^{-1}\left(\frac{v^+_1+v^-_1}{2}\right),1\right)$, and 
    (C5) deals with the maximum with $0$.
    (C6) 
    describes the analogue of 
    (C2)--(C5) for $v_0$.
    Finally, (C7) puts the integer constraints on the various $b_i$ variables.
\end{example}
    
We remark that in particular, we can solve a decision version of finding the threshold (i.e., comparing it to a given bound) in $\mathsf{NP}$.
Additionally, it is not hard to construct games for which the set of functions that satisfy the average property is not convex. This suggests (but does not prove) that formulating the problem in Linear Programming, or indeed finding a polynomial time solution, is unlikely.

\section{Initial budget of exactly the threshold}
\label{sec: what happens at threshold}
Recall from \cref{sec: definitions} that the definition of a threshold function only considers the behavior strictly above or below the threshold. In this section, we study the behavior exactly at the threshold. We present two results. First, surprisingly, we show that when starting with exactly the threshold the game can be \emph{undetermined} (i.e., no player has a winning strategy). Next, we show how to decide in polynomial time whether the threshold in each vertex is 1-strong, 2-strong, or weak. 


\begin{example}
\label{xmp:undetermined}
Consider the game in \cref{fig: example undetermined} with $\lambda=\frac{1}{8}$. Here $v_1$ stands for an initial vertex of some game with a 1-strong threshold of $\frac{7}{18}$ (e.g., the game depicted in \cref{fig: example 7/18}). The only solution to the average property then gives $\Th(v_0)=\frac{5}{18}$. We claim that starting with $\xinit=\Th(v_0)$, the game is undetermined. Indeed, fix a strategy for \PO, we show that \PT can counter it and win. We remind that for a vertex $v$ we have $\Th_\mathrm{diff}(v)=\frac{\Th(v^+)-\Th(v^-)}{2}$, and that the reachability objective allows us to restrict the discussion to memoryless strategies. In $v_0$:
\begin{itemize}[leftmargin=*]
    \item If \PO bids at least $\frac{1}{18}=\Th_\mathrm{diff}(v_0)$, \PT bids 0. \PO wins the bidding with a resulting budget of at most $\Th(v_0^-)=\Th(v_0)=\frac{5}{18}$. If she moves to $v_1$, she has strictly less than the threshold $\Th(v_1)=\frac{7}{18}$ and she loses. If she stays in $v_0$ indefinitely, she also loses.
    \item If \PO bids $\frac{1}{18}-\epsilon$ for $\epsilon>0$, \PT bids $\frac{1}{18}-\frac{\epsilon}{2}$. He wins the bidding and moves to $v_1$, where \PO's budget is $\Th(v_1)-\frac{\epsilon}{2}$, so \PO loses again.
\end{itemize}
Conversely, fix a strategy for \PT, and we show that \PO can counter it and win. In $v_0$:
\begin{itemize}[leftmargin=*]
    \item If \PT bids at least $\frac{1}{18}=\Th_\mathrm{diff}(v_0)$, \PO bids 0. \PT then wins the bidding, and \PO has at least $\Th(v_0^+)=\Th(v_1)=\frac{7}{18}$. If \PT stays in $v_0$ then \PO has strictly more than the threshold and she wins. If \PT moves to $v_1$, \PO still wins since $\Th(v_1)=\frac{7}{18}$ is 1-strong.
    \item If \PT bids $\frac{1}{18}-\epsilon$ for $\epsilon>0$, \PO bids $\frac{1}{18}-\frac{\epsilon}{2}$. She wins the bidding and stays in $v_0$, but increases her budget to $\Th(v_0)+\frac{\epsilon}{2}$, allowing her to win.
\end{itemize}
We conclude that neither player has a winning strategy from $v_0$ with $\xinit=\frac{5}{18}$.
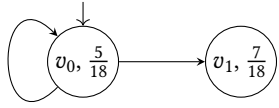
\begin{figure}
\centering  
\begin{tikzpicture}[auto,node distance=2.8cm,scale=0.7]
    \tikzset{every state/.style={minimum size=20pt, inner sep=1},};
    \node (q0) [initial above, state, initial text = {}] at (0,0) {$v_0,\frac{5}{18}$};
    \node (q1) [state] at (3,0) {$v_1,\frac{7}{18}$};
    \path [-stealth]
    (q0) edge [loop, distance=50,out=225,in=135] node {} (q0)
    (q0) edge (q1);
\end{tikzpicture}
\vspace{-0.5cm}
\caption{A game undetermined at the threshold for $\lambda=\frac{1}{8}$}
\label{fig: example undetermined}
\Description{A graph consisting of an initial vertex $v_0$ with a self loop and an edge to $v_1$}
\end{figure}

\begin{figure}
\centering  
\begin{tikzpicture}[auto,node distance=2.8cm,scale=1]
    \tikzset{every state/.style={minimum size=15pt, inner sep=1},};
    \node (q0) [initial above, state, initial text = {}] at (0,0) {$\frac{7}{18}$};
    \node (q1) [state] at (2,0) {$\frac{5}{6}$};
    \node (q2) [state] at (4,0) {$\frac{1}{2}$};
    \node (q3) [state, accepting] at (6,0) {$0$};
    \node (q4) [state] at (3,-0.5) {$1$};
    \path [-stealth]
    (q0) edge (q1)
    (q1) edge (q2)
    (q2) edge (q3) [bend right]
    (q1) edge (q4) [bend left]
    (q2) edge (q4)
    (q0) edge [bend left=15] (q3);
\end{tikzpicture}
\caption{A game with threshold $\frac{7}{18}$ for $\lambda=\frac{1}{8}$}
\label{fig: example 7/18}
\Description{A game played on a DAG whose initial vertex has threshold $\frac{7}{18}$}
\end{figure}
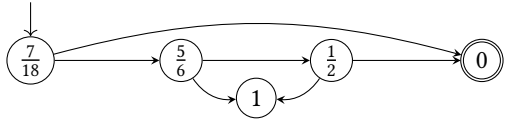
\end{example}

We now turn to show that it is decidable in polynomial time whether the threshold at a vertex is 1-strong / 2-strong / weak. We start with some intuition based on \cref{xmp:undetermined}. There, the reason \PO does not have a winning strategy is that upon bidding $\Th_\mathrm{diff}(v_0)$ and winning, she must go to $v_0^-=v_0$, so no progress is made. Conversely, she must bid $\Th_\mathrm{diff}(v)$ in order not to lose immediately. 
The observant reader may see that if \PO can follow a path consisting only of transitions of the form $(v,v^-)$ until reaching $T$, then she can guarantee winning from the threshold. Indeed, we show that if there is such a path, then the threshold is 1-strong. The latter can be easily checked using graph reachability.

Next, if this condition fails, we need to distinguish between a 2-strong and a weak threshold. We show that this distinction can be made by reduction to a 2-player \emph{turn based} reachability game, which are solvable in polynomial time.
In the following we consider $\lambda>0$. The case of $\lambda=0$ is easier (see \cref{prop:lambda zero threshold is 1 strong iff T reachable}).

\begin{theorem}
    Let $\cG$ be a game with $\lambda>0$. Given the threshold function, it is possible to decide in polynomial time whether the threshold is 1-strong, 2-strong, or weak, for each vertex.
\end{theorem}
\begin{proof}
For simplicity, we add a self loop in each sink, so that $u^-=u^+=u$ are defined for a sink $u$. Note that this does not affect the thresholds.  
We partition the vertices $V$ as follows.
\begin{align*}
    &V_1 =  \{v\in V \mid \Th_\pre(v)<0 \}, \quad 
    V_2 = \{v\in V \mid \Th_\pre(v)>1 \} \\
    &\text{ and }V_\mathrm{mid} =  V\setminus\left(V_1\cup V_2\right).
\end{align*}
Note that due to the self loops on sinks, we have that $T\subseteq V_1$. For $v\in V_1$ we have $\Th(v)=0$, and starting in $(v,\Th(v))$, \PO's budget during the first bidding is
\[
    \WR(0) > \WR(\Th_\pre(v))=\WR(\WRinv(\Th_\mathrm{avg}(v)))=\Th_\mathrm{avg}(v)
\]
so she has strictly more than the threshold in the next vertex. Thus, she can proceed with the standard winning strategy. Therefore, $\Th(v)=0$ is 1-strong. 
A similar argument shows that $\Th(v)=1$ is 2-strong for $v\in V_2$. It is left to deal with $V_\mathrm{mid}$.

We obtain a graph $G_\mathrm{good}$ from $G$ by removing from $G$ all the vertices in $V_2$, as well as every edge $(u,w)$ such that $\Th(w)>\Th(u^-)$. 
Intuitively, all the edges in $G_\mathrm{good}$ are of the form $(u,u^-)$ for every $u^-$ that minimizes the threshold among the neighbors of $u$. To put in the context of \cref{def: average property}, we allow a choice of any of the threshold-minimizing neighbors, instead of fixing one in advance.

For $v\in V_\mathrm{mid}$, if $V_1$ is reachable from $v$ in $G_\mathrm{good}$, then $\Th(v)$ is 1-strong. Indeed, \PO wins from $(v,\Th(v))$ by bidding $\Th_\mathrm{diff}(u)$ at each $u$ and moving along a path to $V_1$, maintaining a budget of exactly the threshold (or losing the bidding and having budget strictly greater than the threshold).

Next, assume $V_1$ is not reachable from $v$ in $G_\mathrm{good}$. We claim that $\Th(v)$ is not 1-strong, and deciding whether it is 2-strong or weak is reducible to a turn-based reachability game.

Starting in $(v,\Th(v))$, fix a strategy $\sigma_1$ for \PO, we show \PT has a strategy $\sigma_2$ such that in the resulting play:
\begin{itemize}[leftmargin=*]
    \item As long as the play stays in $V_\mathrm{mid}$, the configuration $(u,x)$ satisfies $x\leq \Th(u)$.
    \item If the play reaches $V_2$ then $x\leq \Th(u)$.
    \item If the play reaches $V_1$ then $x< \Th(u)$.
\end{itemize}
Using this, since $T\subseteq V_1$, \PO cannot win without leaving $V_\mathrm{mid}$, and since $\Th(u)$ is 2-strong for $u\in V_2$, it follows that \PT wins the play. 

It remains to prove the property above. For $u\in V_\mathrm{mid}$, we have $\Th(u)=\Th_\pre(u)$, and so if the configuration $(u,x)$ satisfies $x\leq \Th(u)$, \PO's budget after WR is $\WR(x)\leq\WR(\Th(u))=\WR(\Th_\pre(u))=\Th_\mathrm{avg}(u)$. We describe \PT's strategy:
\begin{itemize}[leftmargin=*]
    \item If \PO bids at least $\Th_\mathrm{diff}(u)$, then \PT bids 0. \PO then wins the bidding, her new budget is at most $\Th(u^-)$, and the first two conditions follow. Since $V_1$ is not reachable in $G_\mathrm{good}$, we have that any $w\in \neig(u)\cap V_1$ has $\Th(w)>\Th(u^-)$, and so the third condition follows.
    \item If \PO bids $\Th_\mathrm{diff}(u)-\epsilon$ for some $\epsilon>0$, then \PT bids $\Th_\mathrm{diff}(u)-\frac{\epsilon}{2}$. He wins the bidding and moves to $u^+$, at which point \PO's budget is $\Th(u^+)-\frac{\epsilon}{2}<\Th(u^+)$, as needed.
\end{itemize}
It is left to decide whether $\Th(v)$ is 2-strong or weak for $v\in V_\mathrm{mid}$ such that $T$ is not reachable in $G_\mathrm{good}$ from $v$. 
For $u\in V_\mathrm{mid}$, let 
\begin{align*}
    N^+(u)=\{w\in \neig(u) \mid \Th(w)=\Th(u^+)\} \\
    N^-(u)=\{w\in \neig(u) \mid \Th(w)=\Th(u^-)\}
\end{align*}

Fix a strategy for \PT. As long as the play stays in $V_\mathrm{mid}$ and $x=\Th(u)$, if \PT bids $\Th_\mathrm{diff}(u)-\epsilon$ for $\epsilon>0$, then \PO can win by bidding $\Th_\mathrm{diff}(u)-\frac{\epsilon}{2}$ and moving to $u^-$. If \PT bids $\Th_\mathrm{diff}(u)+\epsilon$ then \PO wins by bidding $0$. Therefore, if \PT has a winning strategy, it must prescribe a bid of $\Th_\mathrm{diff}(u)$, and move to some $w\in N^+(u)$, as long as $u\in V_\mathrm{mid}$. 
Restricting \PT to such strategies results in a turn-based game, where at each step \PO can either bid $\Th_\mathrm{diff}(u)$ and move to $N^-(u)$, or bid less and force \PT to move to $N^+(u)$, maintaining the invariant $x=\Th(u)$. Here $V_1$ becomes \PO's target, and $V_2$ becomes a sink. If \PT wins that turn-based game then $\Th(v)$ is 2-strong, and otherwise it is weak.

Since deciding the winner in a turn-based reachability game can be done in polynomial time, we conclude the proof.
\end{proof}



Finally, we show that for $\lambda=0$ things are simpler: if $T$ is reachable then the threshold is 1-strong, and otherwise it is trivially 2-strong. This reduces the decision problem to graph reachability.

\begin{proposition}
\label{prop:lambda zero threshold is 1 strong iff T reachable}
    Let $\cG$ be a game with $\lambda=0$. If $T$ is reachable from $v$ then $\Th(v)$ is 1-strong.
\end{proposition}
\begin{proof}
Since $\lambda=0$, for every $v\in V$ we have $\Th(v^-)\leq\Th(v)$, and equality holds if and only if $\Th(v^-)=\Th(v^+)$.
We obtain from $G$ a graph $G^-$ by removing every edge $(u,w)$ such that $\Th(w)>\Th(u^-)$. 
Intuitively, all the edges in $G^-$ are of the form $(u,u^-)$ for every $u^-$ that minimizes the threshold. To put in context, we allow a choice of any of the threshold-minimizing neighbors, instead of fixing one in advance.

Let $v\in V$ such that $T$ is reachable from $v$ in $G^-$ via some path $P$. 
\PO can now win from $(v,\Th(v))$: At each vertex $u$ she bids $\Th_\mathrm{diff}(u)$. Upon winning the bidding she moves along $P$, thus reaching $T$ if she wins every bidding. Upon losing any bidding, \PO reaches a budget higher than the threshold, and can follow a suitable strategy.

It is left to show that if $T$ is reachable from $v$ in $G$, then it is also reachable in $G^-$. Assume by way of contradiction there exist vertices for which this does not hold, and let $v_0$ be such a vertex with a minimal threshold. 
That is, $T$ is reachable from $v_0$ in $G$ but not in $G^-$, and $v_0$ has minimal threshold among such vertices.

Fix a path $P=v_0,\ldots,v_k$ in $G$. If $\Th(v^-_i)=\Th(v^+_i)$ for all $i<k$ then $P$ is a valid path in $G^-$ as well. 
Otherwise, let $i$ be the minimal index such that $\Th(v^-_i)<\Th(v^+_i)$. Then $\Th(v_j)=\Th(v_0)$ for all $j\leq i$ and the prefix $P'=v_0,\ldots,v_i$ of $P$ is valid in $G^-$. 
Since $\Th(v^-_i)<\Th(v_i)\le 1$, it follows that $T$ is reachable from $v^-_i$ in $G$ (as otherwise the threshold at $v^{-}_i$ would be $1$).
Additionally, $\Th(v^-_i)<\Th(v_i)=\Th(v_0)$, and by the minimality of $\Th(v_0)$ it follows that $T$ is reachable from $v^-_i$ in $G^-$ along a path $P''$. The concatenation $P'P''$ is therefore a path from $v_0$ to $T$ in $G^-$.
\end{proof}

\section{Discussion and Future Research}
\label{sec:discussion}
Robin Hood bidding games incorporate a regulating entity into bidding games, allowing the simulation of realistic settings that cannot be captured with standard bidding games. 
The introduction of wealth redistribution comes at a technical cost: analyzing the behavior of the game becomes much more involved (c.f., \cref{sec: example}). Nonetheless, we are able to show that the model retains the nice property of having a threshold function, albeit the game may become undetermined exactly at the threshold.

Apart from establishing the theoretical and algorithmic foundations of this setting, our results shed light on various properties of the optimal strategies of the players. In particular, we show that when starting above the threshold, \PO intuitively plays on an unwinding of the game to a DAG in order to reach $T$. However, when playing exactly from the threshold, \PO needs a path to $T$ along which she can afford to win all the biddings. 

A natural future direction is to extend our framework to infinite-duration games, e.g., B\"uchi and parity games. For standard bidding games, winning in infinite-duration games reduces to an analysis of strongly connected components~\cite{avni2019infiniteDuration}. In the Robin Hood case, however, this no longer applies, suggesting that showing the existence of a threshold function is nontrivial, if there even exists one.

One view of WR is as a mechanism for changing the budgets of players outside the bidding phase. A different mechanism for achieving this is, introduced in~\cite{avni2024bidding}, designates special vertices where agents can \emph{charge} their budget. From an economical perspective, this can be seen as vertices where a player performs some ``work'' and receives a salary. 
Thus, combining the models would allow us to specify a richer economic dynamics. It would be interesting to examine whether this model retains nice algorithmic properties.

A different research direction concerns viewing wealth redistribution as a form of \emph{discounting}~\cite{broome1994discounting}: in discounting, the value of future rewards decreases exponentially with time, according to some discount factor $\lambda$. Wealth redistribution can then be viewed as a discounting factor on the difference of the budgets of the agents. It may be of interest to consider other models of discounting in bidding games, e.g., a reward model for the agents where future-budgets are worth less than current budgets.

\balance



\begin{acks}
This research was supported by the ISRAEL SCIENCE FOUNDATION (grants No. 989/22 and No. 1679/21).
\end{acks}


\bibliographystyle{ACM-Reference-Format} 
\bibliography{sample}


\end{document}